\documentclass[journal,comsoc]{IEEEtran}
\usepackage[T1]{fontenc}
\usepackage{booktabs}

\ifCLASSINFOpdf
    \usepackage[pdftex]{graphicx}
  \graphicspath{{images/}}
\else
\fi

\usepackage{amsmath}
\usepackage{amsthm}
\theoremstyle{definition}
\newtheorem{definition}{Definition} 
\newtheorem{theorem}{Theorem} 
\begin{document}
\title{Double Link Failure Protection \\ using a Single P-cycle}

\author{Pallavi~Athe, Yatindra~Nath~Singh,~\IEEEmembership{Senior Member,~IEEE.}
\thanks{Pallavi Athe and Yatindra Nath Singh are with the Department of Electrical Engineering,
Indian Institute of Technology, Kanpur, India.}%
\thanks{E-mail:{\{apallavi, ynsingh}\}@iitk.ac.in}}

\maketitle

\begin{abstract}In this letter, we investigate survivability in optical networks for protection from two simultaneous link failures. Failure probability of two links with overlapping protection can be high if these links are geographically close. In a network with deterministic single link protection, simultaneous failure of two links may lead to partial or full loss of traffic on the failed links. Two link failure protection will make the network more resilient by protecting double failures having overlapping protection. A method for achieving double fault tolerance is double cycle method (DB); it uses two pre-configured cycles (p-cycles) to protect a link. Single p-cycle (SG) method, which uses one p-cycle to protect a link from two simultaneous link failure is introduced in this letter. Integer linear programs (ILP) are formulated for the	SG method as well as DB method. It has been observed that	the SG method  provides solution to bigger networks with lesser computational resources as compared to the DB method.
\end{abstract}
\IEEEpeerreviewmaketitle
\section{Introduction}
\IEEEPARstart{O}{ptical}  networks carry enormous amount of information. Failure of any element in an optical network even for a small duration can lead to a large amount of information loss, and consequently, the revenue loss. Thus it is important to build mechanisms of survivability to take care of failures. Survivability of a network is the ability of the network to either maintain uninterrupted flow of information or to minimize the outage period. In an optical network, it is of great importance and has been  studied extensively by the researchers \cite{Zhou} \cite{Habib}.\\
Among various protection schemes, p-cycles are quite promising due to their mesh like efficiency and ring like speed. P-cycles were first introduced by Grover \cite{Grover1} and has been extensively studied for optical network survivability \cite{Asthana}.  A p-cycle can protect working capacity of the on-cycle as well as the straddling links. A straddling link is a chord of a cycle with its end nodes being part of the  cycle. When a straddling link fails, a p-cycle has two paths to reroute the traffic and hence, the capacity requirement in the p-cycle is reduced to half of the protected working capacity in the straddling link. A single copy of a p-cycle can provide unit capacity protection to an on-cycle link and, two unit capacity protection to straddling links on it.\\
Feng et-al \cite{Feng}, described a method  which provides deterministic protection from two link failures using p-cycles. In \cite{Feng}, each link is protected using two p-cycles having link-disjoint protection segments. In this letter, we have described a method which use one p-cycle to protect working capacity on each link against two simultaneous link failures. This protection method is applicable to graphs which has at least three link disjoint paths between each set of nodes like for all other protection methods. Since, our method requires only one p-cycle for each link to protect against two simultaneous link failures, the number of variables involved and the computational time is significantly less. We have also observed that the required spare capacity is more efficient. The efficiency is  even better for optical network with higher average nodal degree.    
\section{SG method}
  Normally a copy of a p-cycle can protect unit capacity on an on-cycle link and two unit capacity on a straddling link. A straddling link has two alternative paths on the p-cycle for restoration, and this attribute of p-cycle is employed in the SG method to protect the optical network from two simultaneous link  failures. In the SG protection method, p-cycles are used to provide protection only to straddling links and no on-cycle protection is used. In case, a p-cycle is shared among multiple straddling links, the number of copies of that p-cycle is taken to be equal to the highest capacity straddling link. Number of copies of p-cycles to protect a straddling link is equal to its capacity. If the capacity of the straddling link to be protected is an odd number, then the number of copies of p-cycle is kept one unit capacity higher than the capacity of the straddling link. Even number of copies of p-cycle are required to ensure protection where the capacity of straddling links sharing the p-cycle are equal and an odd number. For example, if the capacity of every straddling link to be protected is three, then the number of copies of the  p-cycle should be four. Restoration takes place in two different ways depending on which two links have failed \emph{two straddling links}, or \emph{ a straddling and an on-cycle link}. When two straddling links fail simultaneously, the highest capacity straddling link uses half of the capacity of p-cycle, and the other failed link uses the remaining capacity of the p-cycle. When a straddling link and an on-cycle link fail, then the straddling link uses the intact alternative path on p-cycle for restoration. The working paths through the on-cycle links must have been restored as a straddling link by some other p-cycle.
 \begin{theorem}
   A p-cycle $p$ can protect all the straddling links on it from any two simultaneous link failure, if the number of copies of the p-cycle \\
    \[ n_{p}\geq
       \begin{cases}
         W & \quad \text{if $W$ is even},\\ 
        W+1 & \quad  \text{if $W$ is odd.}
       \end{cases}
     \]
  Here, $W$ is the working capacity of the straddling link with maximum capacity, on p-cycle $p$. 
 \end{theorem}
 \begin{proof}
   Consider a p-cycle protecting three straddling links $e1$, $e2$ and $e3$ as shown in Fig. \ref{fig:p1 1}. Assume that $e1$ is the highest capacity straddling link with $w1$ capacity protected by the p-cycle. The number of copy of p-cycle required as per our algorithm to protect network from two link failure will be $w1$ when $w1$ is even, and $w1+1$ when $w1$ is odd. In case of failure of the two straddling links $e1$ and $e2$ as shown in Fig. \ref{fig:p1 1} (a), half of the $w1$(or $w1+1$ when $w1$ is odd) copies of p-cycle will be used to restore $e1$. The link $e2$ is restored using remaining half copies of the p-cycle.\\
    In case of the failure of a straddling link and an on-cycle link as shown in Fig. \ref{fig:p1 1}(b), $w1$ capacity of the intact alternative path on the p-cycle will restore the straddling link $e1$. Link $f$ is restored by some other p-cycle of the network as it must have also been given double fault protection as a straddling link.
 \end{proof}
  \begin{figure}
     \includegraphics[width=6cm]{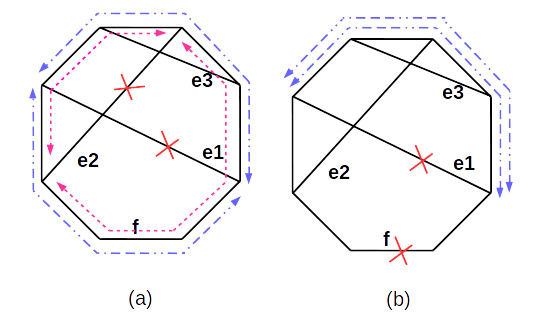}
     \centering
     \caption{ Illustration of two link protection by p-cycle in the SG method. (a) Restoration when two straddling link on p-cycle fail. (b) Restoration when straddling link and an on-cycle link on p-cycle fail.}
     \label{fig:p1 1}
  \end{figure} 
 \subsection{ILP for the SG method}       
 The objective of this ILP is to find the minimum spare capacity required to protect the optical network from two simultaneous link failure using the proposed protection method. We find the minimum spare capacity needed for single fault protection when only straddling link protection is used. Thereafter, number of p-cycles needed are simply doubled to provide double fault tolerance.\\
\textbf{Notations}\\
S=set of links. \\
P=set of cycles.\\
We use $x_{i,p}$ to denote the amount of demand capacity of link $i$ protected by unit capacity of p-cycle $p$ for single fault tolerance with only straddling link protection.\\
  \[ x_{i,p}=
    \begin{cases}
      2   & \quad \text{if link } i \text{ is a straddling link on p-cycle $p$,}\\
      0 & \quad \text{ otherwise.}
    \end{cases}
  \]
We also define, 
  \[ \delta_{i,j}=
  \begin{cases}
  1 & \quad \text {if link $i$ is on cycle $j$,} \\
  0 & \quad \text{otherwise.}
  \end{cases}
  \]
  \textbf{Variables}\\
  $s_{i}$ =spare capacity on link $i$.\\
  $w_{i}$ =working capacity on link $i$.\\
  $c_{i}$=cost of unit capacity of link $i$.\\
  $n_{p}$ = number of copies of p-cycle $p$.\\
  $n_{i,p}$=number of copies of p-cycle $p$ required to protect working capacity of link $i$ in single fault tolerance scenario. \\
  ILP for the SG method is as follows.\\
   \textbf{Minimize:} 
  \begin{equation}\label{eq:1}
  \sum_{i\in S} c_{i}s_{i}.
  \end{equation}
  \textbf{Subject to :}
  \begin{flalign}
 \sum_{{p}\in P} x_{i,p} n_{i,p}\geq w_{i};\quad \quad \forall i\in S \label{eq:2}.\\
          n_{p}\geq 2n_{i,p}; \quad \forall i \in S,  \forall p \in P \label{eq:3}.\\
    s_{i}\geq \sum_{j\in P}n_{j}\delta_{i,j} ;\quad \forall i\in S\label{eq:4}.\\
    n_{p}\geq 0, n_{i,p}\geq 0; \quad \forall i\in S;\forall p \in P\label{eq:13}.
    \end{flalign}
    In equation (\ref{eq:1}), the objective function that minimizes the total spare capacity is defined. Constraint (\ref{eq:2}) ensures that the  protection capacity on p-cycle is sufficient to protect the working capacity on link $i$ as a straddling link under single fault tolerance scenario. Constraint (\ref{eq:3}) selects the minimum number of copies of the p-cycles $p$, required to provide  protection to any two straddling links that fail simultaneously. Constraint (\ref{eq:4}) ensures that the spare capacity on each link is sufficient to form the p-cycles.
 \begin{figure}
    \includegraphics[width=\linewidth]{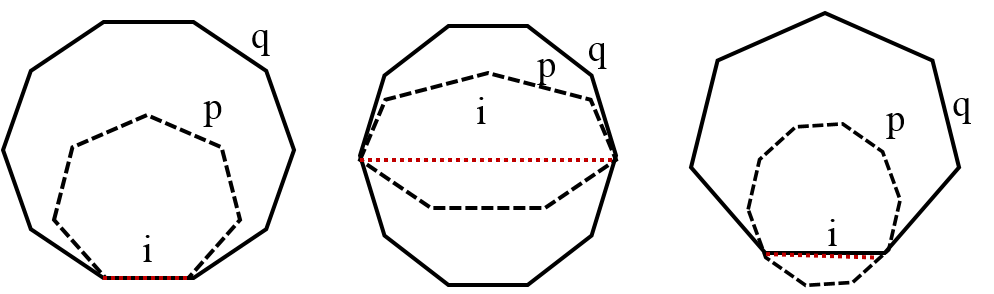}
    \caption{Various possibilities for protection-pair of p-cycles ($p$ and $q$) which can protect the link $i$ from two simultaneous link failure.}
      \label{fig:Picture6 1}
    \end{figure}
\section{DB method}
The most intuitive method for deterministic dual link failure protection is to protect each link of the network with two p-cycles as described in \cite{Feng}. The pair of p-cycles chosen to protect a link must have link disjoint backup paths as shown in Fig. \ref{fig:Picture6 1}. We consider the bidirectional graph G(N,L) and formulate an alternative ILP. This ILP uses a pair of p-cycles to protect each link. Here N is the number of nodes and L is the number of links. The set of cycles, and the set of pair of p-cycles for each link are precomputed and given as inputs to the optimization model.
\begin{definition}[Protection-Pair]
A protection-pair of link $i$, is defined as a pair of p-cycles which are link disjoint or have only one link $i$ as common on-cycle link.
\end{definition}
 \begin{figure}
    \includegraphics[width=6cm]{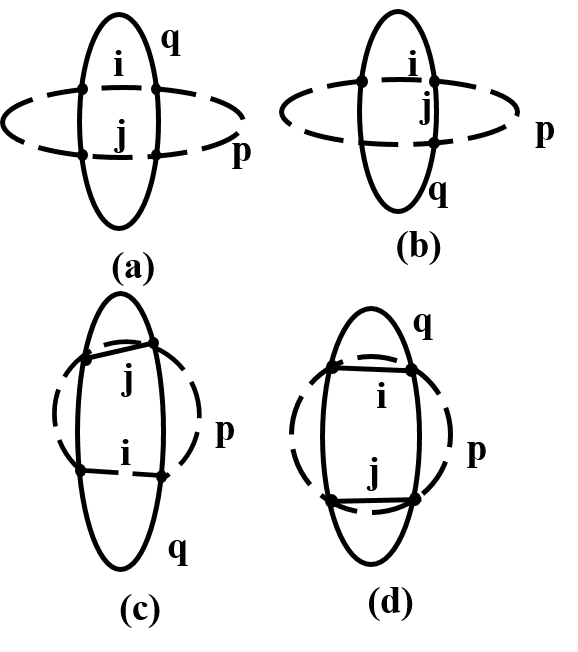}
    \centering
    \caption{Links sharing same pair of protecting p-cycles}
      \label{fig:Picture6 2}
 \end{figure}
The protection-pair for each link is chosen independently. The p-cycles are shared for protecting more than one link. Thus a p-cycle can pair with many other p-cycles simultaneously. In some cases the same pair of protecting p-cycles can protect more than one links. The simultaneous failure of two links having same protection-pair may lead to a situation when, the failed links have only one p-cycle for protection. The spare capacity required in such cases is higher. Possible scenarios when a pair of p-cycles can protect more than one link are shown in Fig. \ref{fig:Picture6 2}.\\
Fig. \ref{fig:Picture6 2} (a) shows the case in which two links $i$ and $j$ are `on-cycle' on p-cycle $p$ and `straddling' on p-cycle $q$. If link $i$ fails, the protection is provided by one of the p-cycles from the pair. If second failure occur on the p-cycle which is used to restore the first failure, then the traffic is diverted on to the unused p-cycle of the pair. Link $j$ is also protected by the pair of p-cycle $p$ and $q$ in the same way. In case of simultaneous failure of link $i$ and $j$, the p-cycle $p$ fails to protect the links, and traffic for both the links is restored by the p-cycle $q$. The required number of copies of p-cycles $p$ and $q$ should be sufficient to incorporate all the two failure scenarios mentioned above. The number of copies of p-cycle $p$ required will be the maximum of the working capacities of $i$ and $j$. The number of copies of p-cycle $q$ will be the sum of half the capacities of $i$ and $j$ approximated to the lowest integer greater than or equal to the sum.
\begin{equation}
n_{q}	=\left \lceil{\frac{w_{i}}{2}+\frac{w_{j}}{2}}\right \rceil
\end{equation}
To ensure protection for the three cases shown in Fig. \ref{fig:Picture6 2} (b), \ref{fig:Picture6 2}(c) and \ref{fig:Picture6 2}(d), number of copies of p-cycle $p$ and $q$ required will be the same as needed to protect the maximum of the capacities on link $i$ and $j$. In the case of simultaneous failure of links $i$ and $j$, the p-cycles $p$ and $q$ protect links $i$ and $j$ respectively.
   
\subsection{ILP for DB method}
This linear program finds the minimum spare capacity required to form p-cycles, ensuring the protection of each link from two simultaneous link failure using the DB method.\\ 
\textbf{Notations}\\
S=set of links. \\
P=set of cycles.\\
$Q_{i}$=set of protection-pair p-cycles of link $i$, indexed as $(p,q)_{i}$.\\
For a link $i$ having $(p,q)_{i}$ in the set $Q_{i}$, we define the indicator variable $x_{i,p,q}$ as follows,
\[ x_{i,p,q} =
 \begin{cases}
1   & \quad \text{if link } i \text{ is on cycle link on p-cycle $p$,} \\
2  & \quad \text{if link } i \text{ is straddling link on p-cycle $p$,} \\
0 & \quad \text{otherwise.}
 \end{cases}
\]
We also define the indicator variable $\delta_{i,j}$,
\[ \delta_{i,j}=
\begin{cases}
 1 & \quad \text {if link $i$ is on-cycle on p-cycle $j$,}\ \\
 0 & \quad \text{otherwise.}
\end{cases}
\]
 \textbf{Variables}\\
$s_{i}$ =spare capacity on link $i$.\\
$w_{i}$ =working capacity on link $i$.\\
$c_{i}$=cost of unit capacity of link $i$.\\
$n_{j}$ = number of copies of p-cycle $j$.\\
$n_{i,p,q}$=number of copies of p-cycle $p$ required to protect link $i$ in pair with p-cycle $q$ to protect link $i$ from two simultaneous failure.\\
$n_{i,p}$=number of copies of p-cycle $p$ required to protect working capacity of link $i$.\\
   \textbf{Minimize:}
 \begin{equation}\label{eq:5}
 \sum_{i\in S} c_{i}s_{i}.
 \end{equation}
 \textbf{Subject to :}
 \begin{flalign}\label{eq:6}
\sum_{(p,q)_{i}\in Q_{i}}[ x_{i,p,q} n_{i,p,q}+ x_{i,q,p} n_{i,q,p}]\geq2w_{i} ;\quad \quad \forall i\in S.
 \end{flalign}
 \text{If link $i$ is an on-cycle or a straddling link on both $p$ and $q$, then}
 \begin{flalign}\label{eq:7}
 n_{i,p,q}= n_{i,q,p}; \quad \forall i\in S; \forall (p,q) \in Q_{i}.
 \end{flalign}
 \text{If link $i$ is a straddling link on $q$ and on-cycle on $p$, then,}
  \begin{flalign}\label{eq:8}
  n_{i,p,q}= 2n_{i,q,p} ;\quad \forall i\in S; \forall (p,q) \in Q_{i}.
  \end{flalign}
 \begin{flalign}\label{eq:9}
   n_{i,p}= \sum_{q\in P, q\neq p} n_{i,p,q} ;\quad \forall p \in P,  \forall (p,q) \in Q_{i}.
 \end{flalign}
 \begin{flalign} \label{eq:10}
          n_{p}\geq n_{i,p}+n_{j,p,q} ;\quad \forall i \in S;  \forall (p,q) \in Q_{j} \cap Q_{i}; 
 \end{flalign}
 \text{\quad \quad  \quad \quad$  i,j \in q ; i \neq j$.}
 \begin{equation}\label{eq:11}
   s_{i}\geq \sum_{j\in P}n_{j}\delta_{i,j} ;\quad \forall i\in S.
 \end{equation} 
 \begin{flalign} \label{eq:12}
 \begin{split}
   n_{p}\geq 0,\quad n_{i,p}\geq 0,\quad n_{i,p,q}\geq 0;\\ \quad \forall i \in S ;\quad\forall p \in P;\quad \forall (p,q) \in Q_{i}.
   \end{split} 
 \end{flalign} 
The objective function (\ref{eq:5}) minimizes the total spare capacity. Constraint (\ref{eq:6}), (\ref{eq:7}) and (\ref{eq:8}) ensure full protection of the working capacity on each link from two simultaneous link failures. Constraint (\ref{eq:9}) calculates the number of copies of p-cycle $p$ required to protect a link $i$ as the sum of the number of copies of the p-cycle $p$ in pair with other p-cycles $q\in P;$ $q\neq p$, required to protect the link $i$. Constraint (\ref{eq:10}) ensures that the number of copies of p-cycle $p$ is sufficient to protect the link with highest working capacity. Constraint (\ref{eq:10}) also takes care of the case when two links share the same pair of p-cycles, as described in this section earlier and illustrated by Fig. (\ref{fig:Picture6 2}). Constraint (\ref{eq:11}) ensures that spare capacity on each link is sufficient to form the p-cycles.
\begin{figure}
            \includegraphics[width=\linewidth]{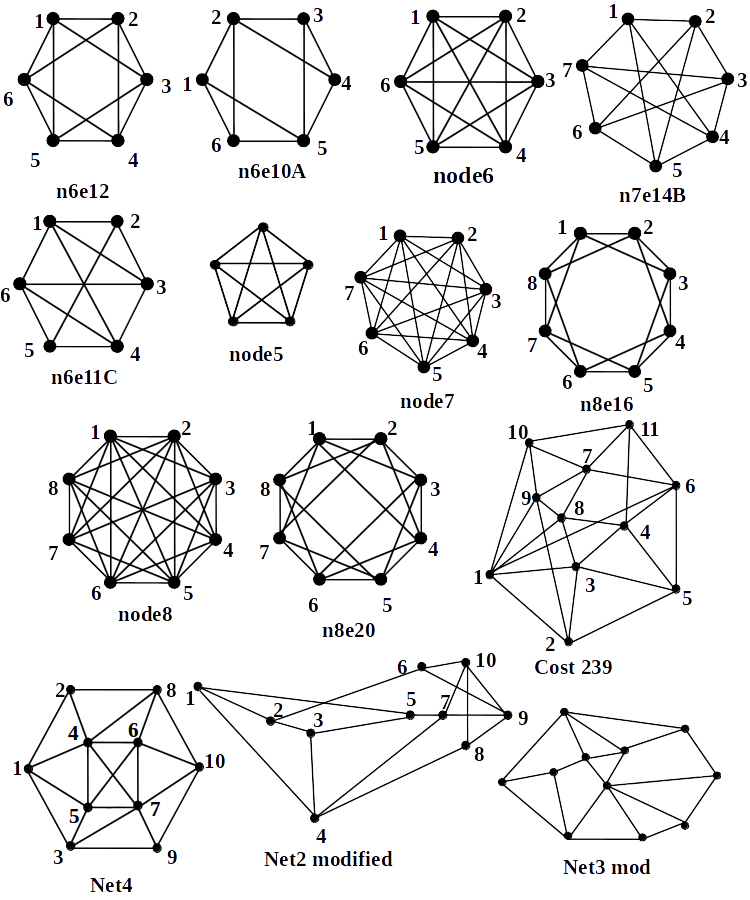}
            \caption{ Network topology used in simulations.}
              \label{fig:Topo3 }
\end{figure} 
   \section{Results and discussion}
ILP formulation is carried out on scilab 5.5, and ILOG CPLEX 9 is used to solve the ILP on AMD Opteron (tm) 1.8 GHz CPU. The number of variables required in the SG method is  $O (L.P)$ where $L$ is the number of links, and $P$ is the number of cycles. In the SG method number of variables and constraints involved in the ILP formulation is  $O(L.P)$. Variables and constraints involved in ILP formulation of DB method is $O(L.P^{2})$. Therefore, the complexity of SG method is lower than the DB method. One can also note that the complexity of the SG method is same as that of the best known single fault protection method. As evident from the simulation results in Table 1, time to solve the ILP is significantly less for the SG method. Spare capacity efficiency (SE) is the ratio of spare capacity required for protection and the working capacity of network, and it is given by Eqn. (\ref{eq:14}).
\begin{equation}\label{eq:14}
	SE = \frac{\sum\limits_{\forall i\in S}s_{i}}{\sum\limits_{\forall i\in S}w_{i}}
\end{equation} 
The SE should be as low as possible. Simulation results shows the smaller value of SE for most of the networks (see table 1) with the SG method as compared to DB method. While for average nodal degree higher than 3.6, SG method always gives better SE.\\
Further, it is observed that the SG method has better restoration speed than DB method. When the first failure occurs, the failed link is restored as a straddling link. In this way, at least half of the copies of p-cycle remain in spare. Restoration path of the first failed link does not require rearrangement if second failure is also on a straddling link. If the second failure occur on the on-cycle link of p-cycle, only half of the traffic requires restoration second time. In case of the DB method restoration of the second failure requires the switching of the traffic from the first p-cycle to a different p-cycle. This restoration will be slower as switching action takes place at nodes which may not be the end nodes of the failed link.

\newcommand{\ra}[1]{\renewcommand{\arraystretch}{#1}}
 \begin{table} \centering
 \ra{1}
 \caption{Simulation Results}
 \begin{tabular}{rrrcrrr}\toprule
 & & & \multicolumn{2}{c}{DB method}  & \multicolumn{2}{r}{Proposed method}\\
  \cmidrule{4-5} \cmidrule{6-7}
  Network & Avg   & Working & SE & ILP time & SE & ILP time  \\
   & degree & capacity &   & (in sec)  &  &  (in sec)  \\
 \midrule 
  n6e10A & 3.3 & 40 & 2.80 & 0 & 2.45 & 0  \\ 
  net2m & 3.4 & 164 & 3.03 & 0.07 & 3.10 & 0.01 \\
  net3mod & 3.4 & 198 & 2.57 & 22.54 & 2.65 & 0.03 \\
  n6e11C & 3.6 & 38 & 2.63 & 0.01 & 2.63 & 0 \\
  n8e16 & 4 & 80 & 1.70 & 115906 & 1.28 & 14.53\\
  n7e14B & 4 & 56 & 1.89 & 18487 & 1.50 & 0.01\\
  n6e12 & 4 & 36 & 1.78 & 32.23 & 1.50 & 0.11\\
  node5 & 4 & 20 & 1.50 & 4.12 & 1.00 & 0\\
  net4 & 4.4 & 142 & \_ & *inf & 1.65 & 0.48\\
  cost239 & 4.7 & 172 & \_ & *inf & 1.06 & 1745.93\\
  node6 & 5 & 30 & 1.00 & 52594 & 0.80 & 2.72\\
  n8e20 & 5 & 72 & \_ & *inf & 1.11 & 0.93\\
  node7 & 6 & 42 & \_ & *inf & 0.67 & 1498.58\\
  node8 & 7 & 56 & \_ & *inf & 0.57 & 195681.18\\
  \bottomrule
\end{tabular}
*inf means solution cold not be achieved by the machine used by us.
\end{table}

\section{Conclusion}
From this study, we can conclude that the spare capacity requirement for deterministic two-link failure protection is less for the SG method as compared to the DB method. SG method is also able to compute the spare capacity of large networks for which DB method fails, because number of variables required for the SG method is $ O(L.P)$ while DB method requires $O(L.P^{2})$ variables. Also, ILP of the SG method is simpler than any other two-link failure protection method. Consequently, the time required in the ILP formulation and ILP solution is also significantly less.
\appendices

\ifCLASSOPTIONcaptionsoff
  \newpage
\fi

\bibliographystyle{IEEEtran}
\bibliography{ref.bib}

\end{document}